\documentclass[10pt]{article}     

\usepackage{allmtt}
\usepackage{graphicx}
\usepackage{amsmath}
\usepackage{amssymb}
\usepackage{latexsym}
\usepackage{stmaryrd}
\usepackage{xspace}
\usepackage{diagrams}
\usepackage{proof2}
\usepackage{url}

\usepackage{amsthm} 
\newtheorem{theorem}{Theorem}
\newtheorem*{theorem*}{Theorem}

\newtheorem{proposition}[theorem]{Proposition}

\theoremstyle{definition}
\newtheorem{definition}[theorem]{Definition}

\theoremstyle{remark}
\newtheorem*{remark}{Remark}

\newcommand{\intersect}{\ensuremath{\cap}}
\newcommand{\Int}{\mathop{\text{Int}}}
\renewcommand{\split}[2]{#1_{{}\ltimes #2}}
\newcommand{\Relax}[1]{\ensuremath{\mathrm{relax}(#1)}}

\newcommand{\vinterp}[1]{\llbracket #1 \rrbracket_e}
\newcommand{\binterp}[1]{\llbracket #1 \rrbracket_!}
\newcommand{\minterp}[1]{\llbracket #1 \rrbracket_p}

\newcommand{\figureline}{\rule{\textwidth}{0.5pt}}
\newcommand{\figureend}{\rule{\textwidth}{0.5pt}}

\newcommand{\iso}{\cong}
\newcommand{\isomorphism}{\cong}

\newcommand{\denote}[1]{
\llbracket #1 \rrbracket} 
\newcommand{\name}[1]{
\ulcorner #1 \urcorner}
\newcommand{\coname}[1]{%
\llcorner #1 \lrcorner}

\newcommand{\sizeof}[1]{
  \left|#1\right|}

\newcommand{\dom}{\operatorname{dom}}
\newcommand{\cod}{\operatorname{cod}}

\newcommand{\catC}{\ensuremath{{\cal C}}\xspace}

\newcommand{\id}[1]{\ensuremath{\mathrm{id}_{#1}}}

\newcommand{\catSet}{
\ensuremath{\textbf{Set}}\xspace}

\newcommand{\inlinegraphic}[2]{
  \dimendef\grafheight=255\dimendef\grafvshift=254
  \grafheight=#1
  \grafvshift=-0.5\grafheight
  \advance\grafvshift by 0.5ex
  \raisebox{\grafvshift}{\includegraphics[height=\grafheight]{images/#2}\xspace}
}

\begin{document}
\title{Graphical Reasoning in Compact Closed Categories for Quantum Computation}

\author{Lucas Dixon\thanks{\texttt{l.dixon@ed.ac.uk}}\\
  University of Edinburgh
  \and 
  Ross Duncan\thanks{\texttt{ross.duncan@comlab.ox.ac.uk}}\\
  University of Oxford 
}


\maketitle

\begin{abstract}
  Compact closed categories provide a foundational formalism for a
  variety of important domains, including quantum computation.  These
  categories have a natural visualisation as a form of graphs.  We
  present a formalism for equational reasoning about such graphs and
  develop this into a generic proof system with a fixed logical kernel
  for equational reasoning about compact closed categories.
  Automating this reasoning process is motivated by the slow and error
  prone nature of manual graph manipulation. A salient feature of our
  system is that it provides a formal and declarative account of
  derived results that can include `ellipses'-style notation. We
  illustrate the framework by instantiating it for a graphical
  language of quantum computation and show how this can be used to
  perform symbolic computation.  
\end{abstract}

\noindent Keywords: graph rewriting, quantum
    computing, categorical logic, interactive theorem proving,
    graphical calculi, ellipses notation.

\section{Introduction}
\label{sec:introduction}

Recent work in quantum computation has emphasised the use of graphical
languages motivated by the underlying logical structure of quantum
mechanics
itself~\cite{AbrCoe:CatSemQuant:2004,Selinger:dagger:2005,Coecke2005Kindergarten-Qu,Coecke2006POVMs-and-Naima,Coecke2006Quantum-Measure}.
These techniques have a number of advantages over the conventional
matrix-based approach to quantum mechanics:

\begin{itemize}
\item The visual representation abstracts over the values in the
  matrices. This removes detail that is difficult or tedious for a
  human to interpret. 

\item Many properties have a natural graphical representation. For
  example, separability of quantum states can be inferred from disjoint subgraphs.

\item The algebra of graphs generalises to domains other than
  vector spaces:  it provides a representation for
  compact closed categories~\cite{KelLap:comcl:1980}. 

\end{itemize}

A major problem with these graphical representations is the lack of
machinery for automating their manipulation. The main contribution of
this paper is a graph-based formalism that is suitable for
representing and reasoning about compact closed categories with
additional equational structure. This has a wide variety applications
including reasoning about relations, stochastic processes, and
synchronous processes. In this paper we introduce the representation,
develop it into a formal proof system, and highlight its application
for symbolic reasoning about quantum computation.

We begin by presenting a graphical model of quantum computation.  This
model displays the typical features of a graphical calculi: quantum
processes are represented by graphs built up from basic elements.
Non-structural equivalences are captured by equations between graphs.
An important result in this calculus is the Spider Theorem which takes
the form of an equation between graphs involving informal ellipses
notation (see \S\ref{sec:informal-spider}).

The formalisation, in a graphical form, of rules containing ellipses
notation and the corresponding reasoning with such rules requires an
extension of the graphical calculus that eventually forms {\em graph
  patterns}. We develop this by first defining a formalism for graphs,
their transformations, and an appropriate subgraph relation.  An
important difference between this approach and standard texts on graph
theory lies in the notion of subgraph. We view vertices as operations
which have types corresponding to their incident edges and thus we do
not allow additional edges in a subgraph.  We introduce a general form
of graph combination, called \emph{plugging}, which includes both
parallel and sequential composition as special cases.  Since redexes
are preserved by plugging, this gives a compositional account of
equational reasoning for compact closed categories. We prove that our
graph-based formalism is a faithful representation the of free compact
closed category generated from its basic elements. We also introduce a
general formalism for ellipses notation in graphs which forms {\em
  !-box graphs}. By combining {\em !-box graphs} with our
compositional graph formalism, we provide a suitable representation
for {\em graph patterns} that can formally represent and reason with
rules derived from the Spider Theorem.

Using our graph-based formalism as the representational foundation, we
develop a simple logical framework for manipulating models of compact
closed categories. This has a suitable rewriting mechanism where the
axioms of the underlying object-formalism are expressed as equations
between graphs. We then present a short case study that illustrates
the framework by instantiating it for the introduced model of quantum
computation.  This shows how the framework can be used to symbolically
perform simplifications of quantum programs as well as simulate
computations.







\section{Quantum Computations as Graphs}
\label{sec:exampl-quant}

In this section we will describe a set of generators and equations
used to reason about quantum computation, and show how some of its
formal properties lead to particular issues for the development of
reasoning machinery.

Initiated in~\cite{AbrCoe:CatSemQuant:2004}, a substantial strand of
work in quantum informatics has involved the development of high-level
models of quantum processes based on compact closed categories. In
these formalisms, quantum processes---such as quantum logic gates, or
the measurement of a qubit---correspond to arrows in the category,
while the different quantum data types, usually just arrays of qubits,
are the objects. 

In terms of the graphical language, a recent account is described by
Coecke and Duncan~\cite{Coecke:2008jo}. This is based on providing a
graphical language for compact closed categories, described
in~\S\ref{sec:comp-clos-categ}. This account allows edges to represent
qubits and, in particular, the domain and codomain edges represent the
inputs and outputs respectively of a quantum process\footnote{In this
  account, no interpretation of edge direction is needed as objects in
  the underlying categorical model are self dual.}.  Internally,
several edges may represent the same physical qubit at different
times. An edge may even represent a ``virtual'' qubit which stands for
a correlation between different parts of the system.  Coloured nodes
(a light green and a darker red) are used to denote two families of
operations on qubits, expressed graphically as the following
generators:

\begin{gather*}
  \epsilon_Z = \inlinegraphic{1.5em}{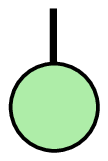} \qquad
  \delta_Z = \inlinegraphic{1.5em}{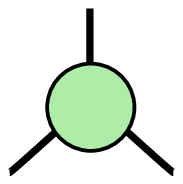} \qquad
  \epsilon_Z^\dag = \inlinegraphic{1.5em}{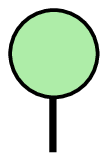} \qquad
  \delta_Z^\dag = \inlinegraphic{1.5em}{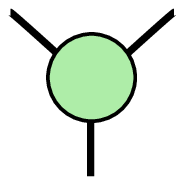} \qquad
  \alpha_Z = \inlinegraphic{1.5em}{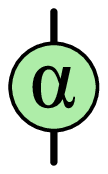} \qquad
\\  
  \epsilon_X = \inlinegraphic{1.5em}{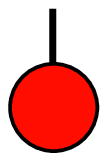} \qquad
  \delta_X = \inlinegraphic{1.5em}{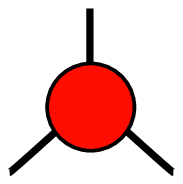} \qquad
  \epsilon_X^\dag = \inlinegraphic{1.5em}{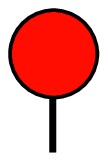} \qquad
  \delta_X^\dag = \inlinegraphic{1.5em}{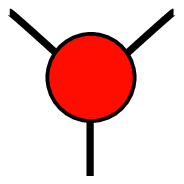} \qquad
  \alpha_X = \inlinegraphic{1.5em}{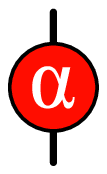} \qquad
\end{gather*}

\noindent where $\alpha \in [0,2\pi)$. The $\delta_Z$ and $\epsilon_Z$
represent quantum operations which respectively copy and delete the
eigenstates of the Pauli $Z$ operator.\footnote{Uniform copying
  operations are forbidden by the no-cloning theorem
  \cite{Wootters1982A-single-quantu}, but such operations are possible
  if we demand only the eigenstates of some self-adjoint operator to
  be copied.  Other states will not not copied.  The same remarks hold
  true for erasing \cite{Pati2000Impossibility-o}.} In addition, we
have \inlinegraphic{1.5em}{} which represents a Hadamard gate.
Notice that $\delta_Z$ has one edge in its domain for the qubit to be
copied, and two edges in its codomain for the two copies it produces.
Similarly, $\epsilon_Z$ has one qubit as input and no outputs.  The
adjoints $\delta^\dag_Z$ and $\epsilon^\dag_Z$ correspond to an
operation known as \emph{fusion}, and to the operation of preparing a
fresh qubit in a certain state.  The $\alpha_Z$ corresponds to phase
shift of angle $\alpha$ in the $Z$ direction. The family of maps
indexed by $X$ are defined in exactly the same way, but relative to
the Pauli $X$ operator rather than the $Z$.

The free compact closed category is then given by all graphs formed by
composing and tensoring these basic graphs. All quantum operations may
be defined by combining these simple operations---which are
essentially classical---on two complementary observables.

We emphasise that this is a notation for representing quantum
processes, not just quantum states. In this setting a state is simply
a process with no inputs; that is, a concrete graph with empty domain.
Since our formalism is based on the underlying mathematical structure
rather than any particular model of quantum computation, it is capable
of representing quantum circuits, measurement-based quantum
computations, as well as other models. Indeed, an important
application of this work is to show that states or computations
implemented differently are equivalent.

The beauty of graphical calculi for compact closed categories is that
equations which hold for general algebraic reasons are absorbed into
the notation.  However in order to represent 
quantum computation, generic structure will not suffice: we need
additional equations between graphs. In the system we present here,
these describe the interaction between complementary 
observables and allow equivalent computations to be proved equivalent.
The equations are discussed in detail in \cite{Coecke:2008jo} and are
presented here graphically in Figure~\ref{fig:graph-quant-eqns}.

\begin{figure}[ptbh]
  \figureline
  \begin{description}
  \item[Comonoid Laws] 
    \[
    \begin{array}{ccccccccccccccc}
      \inlinegraphic{2.5em}{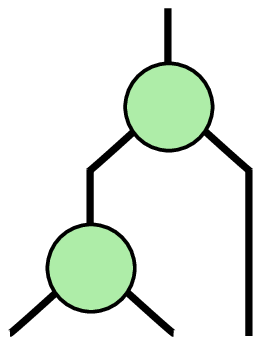}  &=\;&
      \inlinegraphic{2.5em}{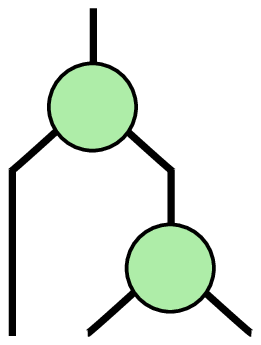}
      &\qquad&
      \inlinegraphic{2.5em}{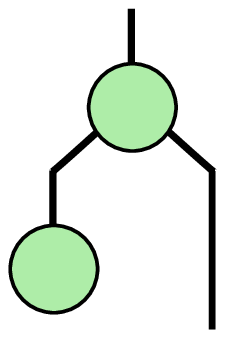}  &=\;&
      \inlinegraphic{2.5em}{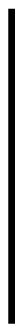}  &=\;&
      \inlinegraphic{2.5em}{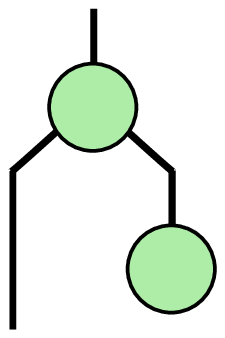} 
      &\qquad&
      \inlinegraphic{2.5em}{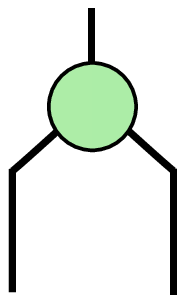}  &=\;&
      \inlinegraphic{2.5em}{}
    \end{array}
    \]
  \item[Isometry, Frobenius, and  Compact Structure]
    \[
    \begin{array}{cccccccccccccccc}
      \inlinegraphic{2.5em}{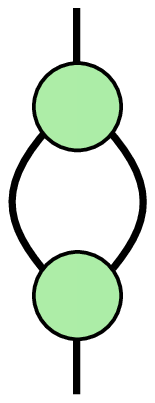}  &=\;&
      \inlinegraphic{2.5em}{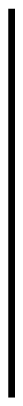}  
      &\qquad\qquad &
      \inlinegraphic{2em}{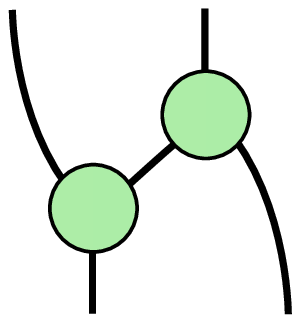}  &=\;&
      \inlinegraphic{2em}{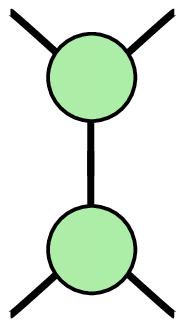}  
      &\qquad\qquad&
      \inlinegraphic{2em}{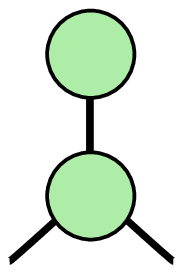}  &=\;&
      \inlinegraphic{1.5em}{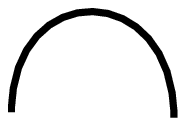}  
    \end{array}
    \]
  \item[Abelian Unitary Group and Bilinearity]
    \[
    \begin{array}{ccccccccccccccccccccccccc}
      \inlinegraphic{2.2em}{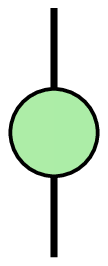}  &:=\;&
      \inlinegraphic{2.2em}{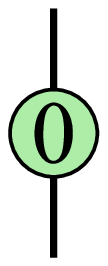}  &=\;&
      \inlinegraphic{2.2em}{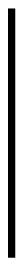}  
      &\qquad\quad&
      \inlinegraphic{3.5em}{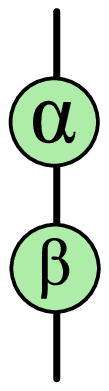}  &=\;&
      \inlinegraphic{3.5em}{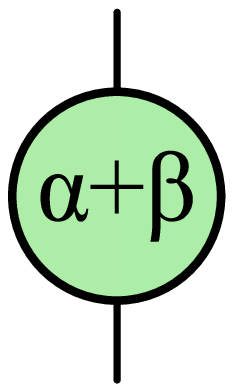}  &=\;&
      \inlinegraphic{3.5em}{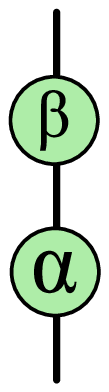} 
    \end{array}
    \]
    \[
    \begin{array}{ccccc}
      \inlinegraphic{3.5em}{}  &=\;&
      \inlinegraphic{3.5em}{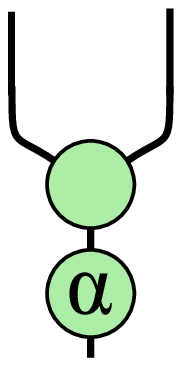}  &=\;&
      \inlinegraphic{3.5em}{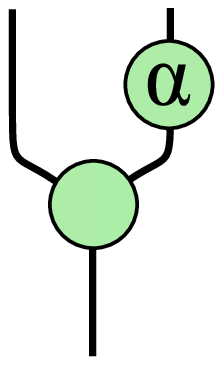} 
    \end{array}
    \]
  \item[Bialgebra Laws] Let  $\inlinegraphic{1em}{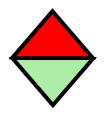} := 
    \inlinegraphic{1.5em}{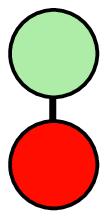}$;  then:
    \[
    \begin{array}{ccccccccccccccc}
      \inlinegraphic{3.5em}{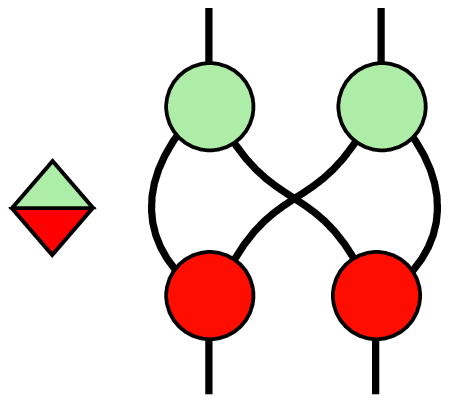}  &=\;&
      \inlinegraphic{3em}{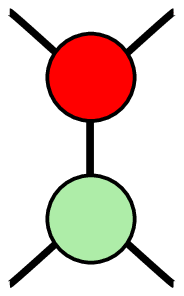}
      &\qquad\quad&
      \inlinegraphic{2.5em}{}  &=\;&
      \inlinegraphic{1.5em}{}
  \end{array}
  \]
\item[Group Actions]
  \[
  \begin{array}{ccccccc}
    \inlinegraphic{2.5em}{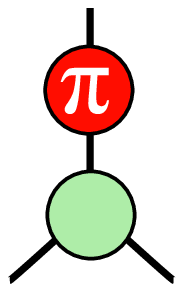} &=\;&
    \inlinegraphic{2.5em}{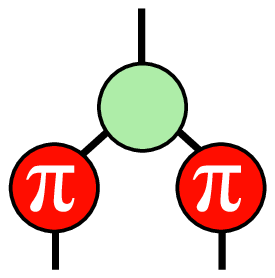} 
    &\qquad&
    \inlinegraphic{2.5em}{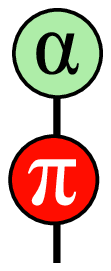} &=\;&
    \inlinegraphic{1.5em}{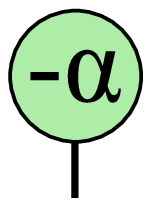} 
    \\
    \\
    \inlinegraphic{2.5em}{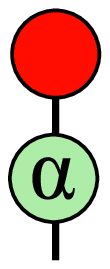} &=\;&
    \inlinegraphic{1.5em}{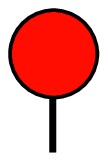} 
    &\qquad&
    \inlinegraphic{2.5em}{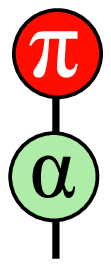} &=\;&
    \inlinegraphic{1.5em}{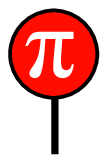} 
  \end{array}
  \]
\item[$H$ Property and Colour Duality]
  \[
  \begin{array}{cccccccccccccccc}
    &&&&
    \inlinegraphic{2.5em}{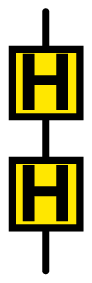} &=\;&
    \inlinegraphic{2.5em}{}
    \\
    \inlinegraphic{2.5em}{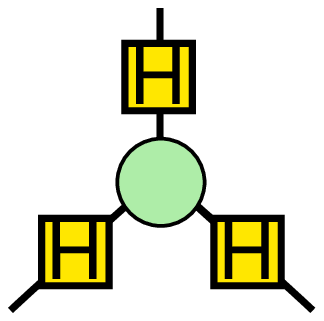} &=\;&
    \inlinegraphic{1.5em}{reddelta}
    &\qquad\quad&
    \inlinegraphic{2em}{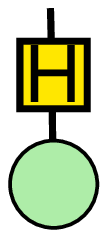} &=\;&
    \inlinegraphic{1.5em}{redepsilon}
    &\qquad\quad&
    \inlinegraphic{3.3em}{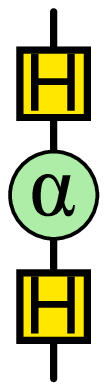} &=\;&
    \inlinegraphic{1.5em}{redalpha}
  \end{array}
  \]
\end{description}  \vspace{-1em}
\caption{Graphical Equations for Quantum Systems. In addition, we have
  a ``colour duality'': each equation shown here gives rise to second,
  which is obtained by exchanging the two colours. The colour duality
  is derivable from the equations involving $H$.}
  \label{fig:graph-quant-eqns}
  \figureend
\end{figure}

The equations from Figure \ref{fig:graph-quant-eqns} which involve
only one colour allow the remarkable \emph{spider theorem}, first
noted in \cite{Coecke2006Quantum-Measure}, to be proved:

\begin{theorem}[Spider Theorem]
  Let $G$ be a connected graph generated
  from $\delta_Z$, $\epsilon_Z$, $\alpha_Z$ and their adjoints; then
  $G$ is totally determined by the number of inputs, the number of
  outputs, and the sum modulo $2\pi$ of the $\alpha$s which occur in
  it.
\end{theorem}

\noindent Hence any connected subgraph involving nodes of only one
colour may be collapsed to a single vertex, with a single value
$\alpha$, giving a ``spider''. Conversely, a spider may be arbitrarily
divided into sub-spiders, provided the total in- and out-degree is
preserved, along with the sum of the $\alpha$s. Informally, this can
be depicted graphically as the equation:

$$\begin{array}{ccc}
\inlinegraphic{6em}{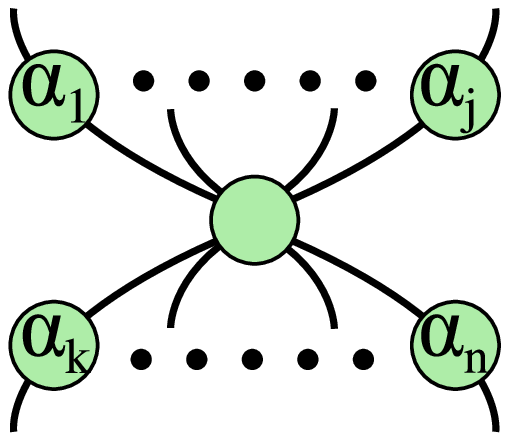} & = & \inlinegraphic{6em}{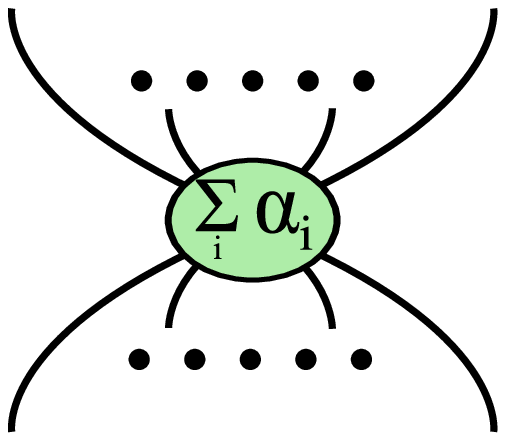}
\end{array}$$
\label{sec:informal-spider}


\noindent From this one can derive $n$-fold versions of many of the
other equations.

%

Spiders offer a very intuitive way to manipulate graphs, and are far
more compact and convenient in calculations than the graphs built up
naively from the generators.  However, no finite set of equations
suffices to formalise spiders: we must move from finite graphs, where
each vertex has bounded degree, and which are subject to a finite
number of equations, to a system where nodes may have arbitrarily many
edges, and there are infinitely many concrete equations. The desire to
retain intuitive reasoning methods for these infinite families of
equations motivates the extension from concrete graphs to \emph{graph
  patterns}, the main subject developed in this paper.

\section{Graphs}
\label{sec:graphs}

\begin{definition}[Graph]
  A \emph{directed graph}\footnote{ Equivalently: a directed graph is
    a functor $G$ from $\bullet \pile{\rTo\\\rTo} \bullet$ to \catSet;
    a graph morphism is then a natural transformation $f: G
    \Rightarrow H$.  } consists of a 4-tuple $(V,E,s,t)$ where $V$ and
  $E$ are sets, respectively of \emph{vertices}\footnote{We will use
    the words ``vertex'' and ``node'' interchangeably.} and
  \emph{edges}, and $s$ and $t$ are maps which give the source and
  target vertices of a an edge respectively:
  \begin{diagram}
    E & \pile{\rTo^{\qquad s\qquad}\\\rTo_t} & V
  \end{diagram}
\end{definition}

\noindent We will assume throughout this paper that both $V$ and $E$
are finite.

\begin{remark}
  Note that any number of edges are allowed between vertices,
  including from a vertex to itself.
\end{remark}

Let $\text{in}(v) := t^{-1}(v)$ and $\text{out}(v) := s^{-1}(v)$ denote the
\emph{incoming} and \emph{outgoing} edges at a vertex $v$.  The
\emph{degree} of a vertex $v$ is $\deg(v) := \sizeof{\text{in}(v)} +
\sizeof{\text{out}(v)}$. To distinguish between elements of different
graphs, we will use the subscript notation $G = (V_G,E_G,s_G,t_G)$.

We say that a vertex $v$ is a \emph{successor} of $u$ if the there
exists an edge $e$ such that $s(e) = u$ and $t(e) = v$.  A pair of
vertices are \emph{connected}, written $u \sim v$,  if they lie in the
reflexive, symmetric, transitive closure of the successor relation.  The
equivalence classes $V/\!\!{}\sim$ are the \emph{connected components}
of $G$.  We write $\sizeof{v}$ to denote the equivalence class
containing the vertex $v$;  we write $[v]$ to denote the subgraph
determined by $\sizeof{v}$.

\begin{definition}[Graph Morphism]
  Given graphs $G$ and $H$, a \emph{graph morphism} $f : G\to H$
  consists of functions $f_E : E_G \to E_H$ and $f_V:V_G\to V_H$ such
  that:
  \begin{gather}
    s_H\circ f_E = f_V \circ s_G,\label{eq:graph-hom1}\\
    t_H\circ f_E = f_V \circ t_G\label{eq:graph-hom2}.
  \end{gather}
\end{definition}

\noindent These conditions ensure that the structure of the graph is
preserved.




\begin{definition}[Open Graph, Open Graph Morphism]\label{def:open-graph}
  An \emph{open graph} $\Gamma = (G, \partial G)$ consists of a
  directed graph $G$, and a set of vertices $\partial G \subseteq V_G$,
  such that for each $v \in \partial G$ we have $\deg(v) = 1$.  The
  set $\partial G$ is called the \emph{boundary} of $\Gamma$; 
 those vertices in $V_G \setminus \partial G$ are called the
\emph{interior} of $\Gamma$, written $\Int G$.

Given open graphs $(G,\partial G)$ and $(H, \partial H)$ a graph
morphism $f : G\to H$ defines a \emph{morphism of open graphs} 
$f: (G,\partial G) \to (H, \partial H)$ if $f_V(v) \in \partial H
\Rightarrow v \in \partial G$ for all $v$ in $V_G$.
\end{definition}

\noindent We will refer to an open graph $(G,\partial G)$
simply as $G$ when it is unambiguous to do so.

\begin{definition}[Strict Map]
  Let $f: (G,\partial G) \to (H, \partial H)$ be an open graph morphism
  say that $f$ when \emph{strict} if $\forall e \in E_H$, if
  $s_H(e) \in f_V(\Int G)$ or $t_H(e) \in f_V(\Int G)$ then $\exists e' \in
  E_G$ such that $f_E(e') = e$.
\end{definition}
\noindent Strictness ensures that there are no additional edges
connected to vertices in the image of $\Int G$.

We emphasise two points about the distinction between interior and
boundary nodes for open graphs.  We view graphs as computational
objects, built up by connecting smaller objects together; we view the
interior vertices as computational primitives.  Strict maps ensure
that the interior structure---the types and connections of the
vertices---is preserved.  The boundary of an open graph defines the
\emph{interface} of the system; the boundary nodes indicate
this interface, and do not carry computational meaning.  Hence
boundary nodes have degree one: they simply mark an edge where
something may be connected.  Morphisms of open graphs preserve this
view by not allowing interior nodes to be mapped to the boundary.

We can also view graphs as topological spaces.  In this
case the boundary nodes can be seen as points which lie outside the
space but are needed to define it, similar to the end points
of an open interval.  From this point of view, morphisms of open
graphs are \emph{continuous}.  What then are the open sets of this space?
Open subgraphs arise via two graph
operations: removing connected components from the graph, and removing
single points.  We note that it suffices to consider removing points
which lie on edges,  since vertex removal can be simulated by
disconnecting the  vertex and then removing the resultant component.
Since we are indifferent to which point on the edge is removed, we
introduce the notion of \emph{splitting an edge}.  The intuition is
that by removing a point from the middle of the edge $e$, we introduce two
new boundary points.  

\begin{definition}[Splitting an Edge]
\label{def:split}
Let $G$ be an open graph, and suppose $e\in E_G$; we define
$\split{G}{e}$, the \emph{splitting of $G$ on $e$}, via the 
graph $G' = (V_G+\{e_1,e_2\},(E_G \setminus \{e\})+ \{e_1,e_2\}, s',t')$,
where $e_1, e_2$ do not occur in $V_G$ or $E_G$, and $s'$ and $t'$ are
defined such that 
\begin{itemize}
\item $s'(e_1) = s_G(e)$, $t'(e_1) = e_1$; 
\item $s'(e_2) = e_2$, $t'(e_2) = t_G(e)$;
\end{itemize}
and they otherwise agree with $s_G$ and $t_G$ respectively.  Then
$\split{G}{e} := (G', \partial G +\{e_1, e_2\})$. 
\end{definition}

\noindent We define a canonical morphism $i$ embedding $\split{G}{e}$ back into
$G$ as follows:
\begin{itemize}
\item $i_V(e_1) = t_G(e_1)$;  $i_V(e_2) = s_G(e_2)$; and $i_V(v) = v$ otherwise.
\item $i_E(e_1) = i_E(e_2) = e$; and $i_E(e') = e'$ otherwise.
\end{itemize}
Clearly, $i$ is injective on the portion of $\split{G}{e}$ excluding
$e_1$ and $e_2$,  and it is strict.

\begin{definition}[Removing a component]
\label{def:remov-component}
  Let $\Gamma = (G,\partial G)$ be an open graph, and suppose that
  $v\in V_G$.  The graph obtained by removing the component
  $[v]$ is denoted $\Gamma - [v] :=
  (G-[v], \;\partial G \setminus  (\sizeof{v} \intersect \partial
  G)$ where the underlying graph is given by:
  \[
  G - [v] = 
  (V_G \setminus \sizeof{v}, \;\;E_G\setminus s_G^{-1}(\sizeof{v}),\;\;
  s_G|_{(E_G\setminus s_G^{-1}(\sizeof{v}))}, \;\;t_G|_{(E_G\setminus s_G^{-1}(\sizeof{v}))}\;).
  \]
\end{definition}

\noindent Writing $G+H$ for the disjoint union of open graphs, it is
immediate that  we have the isomorphism $G  \iso [v] + (G - [v])$, and
hence that the coproduct injection $\text{in}_2 : (G - [v]) \rInto{}{} {[v]}
+ (G - [v])$ provides a canonical map back into the 
original graph.  A further consequence is that every graph is equivalent to
the disjoint union of its connected components.


It is easy to show that the operations of
splitting edges and removing components generalise  to sets of edges
and vertices, and further that any sequence of such operations can be
standardised so that  all the  splittings come first.

\begin{definition}[Open Subgraph]
  Let $G$ be an open graph; then each pair $(F,U)$ with $F \subseteq
  E_G$ and $U\subseteq V_G$ defines an \emph{open subgraph}
  $\split{G}{F}-[U]$. 
\end{definition}

Every open subgraph of $G$ has a canonical map embedding it back into
$G$, constructed from the canonical embeddings at each step; it is
strict, and injective everywhere except the new edges and boundary 
nodes introduced by splittings.

\begin{definition}[Exact Embedding]
\label{def:exact-embedding}
We call an open graph morphism 
\[
f: (G,\partial G) \rInto (H, \partial H)
\]
an \emph{exact embedding}  if:
\begin{enumerate}
\item $f$ is strict;
\item $f_E$ is injective;
\item $f_V$ is injective; and, 
\item $f_V(v) \in \partial H \Leftrightarrow v \in \partial G$, for
  all $v \in V_G$..
\end{enumerate}
\end{definition}




\begin{definition}[Matching]
  We say that $G$ \emph{matches} $H$ if there exists an open subgraph $H'$ of
  $H$, and an exact embedding $e:G\rInto H'$.  In this case we write
  $G \leq H$;  we write $\denote{G}$ for the set of all graphs which
  $G$ matches.
\end{definition}

\begin{proposition}
  Let $G$, $H$, and $K$ be open graphs.  Then 
  \begin{enumerate}
  \item $G \leq$ G;
  \item $G\leq H$ and $H \leq K$; then $G \leq K$;
  \item If $G \leq H$ and $H \leq G$ then $G \iso H$.
  \item $G \leq H$ iff $\denote{H} \subseteq \denote{G}$.
  \end{enumerate}
\end{proposition}
\begin{proof}
  The first property follows from the fact that the identity map is
  an exact embedding;  the second and fourth hold because exact embeddings
  are closed under composition.  For the third property: since we can
  exactly embed $G$ into a subgraph of $H$, and vice versa, we must have
  that these subgraphs are isomorphic to the original graphs; from
  here the isomorphism between G and H is easily
  constructed.
\end{proof}


\section{Graphs with Exterior Nodes}

We now present a generalisation of the open graphs described in the
previous section.  The purpose of this generalisation is to offer more
precise control over matching:  a graph $G$ will match $H$ when it can be
exactly embedded \emph{in a given configuration}.

\begin{definition}
  An \emph{extended open graph}, henceforth abbreviated
  \emph{e-graph},  is pair $(G,X)$ where $G$ is a graph and $X
  \subseteq V_G$ is a distinguished set of vertices.  The elements of
  $X$ are called the \emph{exterior nodes} of $G$;  those vertices in
  $V_G \setminus X$ are called the \emph{interior}.

  An e-graph morphism $f:(G,X)\to (H,Y)$ is a graph morphism such that
  $f_V(v) \in Y$ implies $v \in X$ for all $v\in V_G$.
\end{definition}

The exterior nodes of an e-graph generalise the boundary nodes of an
open graph and are viewed in the same way:  as points outside the
graph.  As well as marking the edge of the graph, exterior points also
constrain how the edges incident at them may be embedded into a larger
graph: the must meet at the same point.  This will be made explicit below.

\begin{definition}[Splitting a Vertex]
  Let $(G,X)$ be an e-graph with $x\in V_G$; we define a new e-graph
  $\split{G}{x}$ by \emph{splitting the  vertex $x$} as $\split{G}{x}
  := (G', (X\setminus \{x\}) + \mathrm{in}(x) + \mathrm{out}(x))$ where 
  $G' := ((X\setminus \{x\}) + \mathrm{in}(x) + \mathrm{out}(x)), E_G, s',
  t')$ and
  \begin{gather*}
    s'(e) = e \quad \text{if } e\in \mathrm{out}(x), \\
    t'(e) = e \quad \text{if } e\in \mathrm{in}(x), \\
    s'(e) = s_G(e), \quad t'(e) = t_G(e) \text{ otherwise.}
  \end{gather*}
\end{definition}
We can define a canonical map $i:\split{G}{x} \to G$ by $i_E = \id{}$, and
$i_V(v) = x$ if $v\in \mathrm{in}(x)+\mathrm{out}(x)$ and $i_V(v) = v$
otherwise.  Evidently, the splitting operation can be lifted to sets
of vertices, so we may write $\split{G}{U}$ when $U\subseteq V_G$.
We define a relation $\heartsuit$ over the vertices of $\split{G}{U}$ by
$v_1 \heartsuit v_2$ iff $i_V(v_1) = i_V(v_2)$.  


\begin{definition}[Relaxtion of an Extended Graph]
  Let $(G,X)$ be an e-graph; define its \emph{relaxation},  $\Relax{G}
  := \split{G}{X}$.
\end{definition}

Essentially $\Relax{G}$ is the open graph that matches $G$ the
closest.  Note that if $G$ is an open graph itself---i.e. all its
exterior points are of degree one---then $\Relax{G} \iso G$.  

\begin{definition}[Matching an Extended Graph]
  We say that $(G,X)$ matches $(H,Y)$ when there exists $H'$, an open
  subgraph of $\Relax{H}$ and an exact
  embedding $f:\Relax{G} \to H'$ such that if $v\heartsuit u$ in
  $\Relax{G}$ then $f(v) \heartsuit f(u)$ in $\Relax{H}$.  In this case we
  write $G \leq_e H$.  As before we define $\denote{G}_e := \{ H | G
  \leq_e H\}$
\end{definition}

\begin{proposition}
\label{match-interp-thm}
$G \leq_e H \Leftrightarrow \denote{G}_e \supseteq \denote{H}_e$.
\end{proposition}

\begin{figure}[t]
  \begin{tabular}{ccccc}
  \scalebox{0.5}{\includegraphics{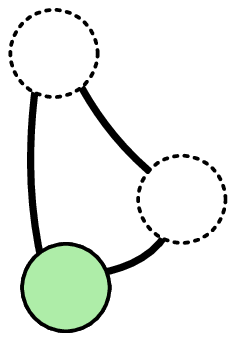}} & 
  \scalebox{0.5}{\includegraphics{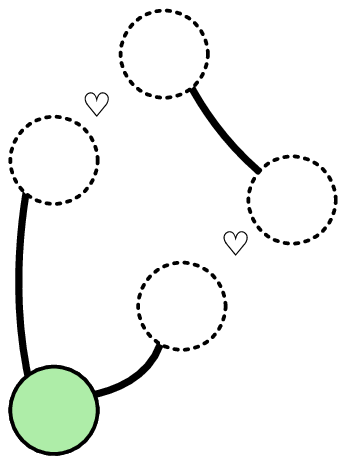}} & 
  \scalebox{0.5}{\includegraphics{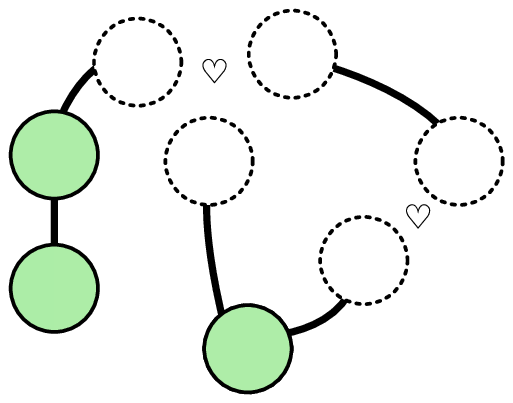}} & 
  \scalebox{0.5}{\includegraphics{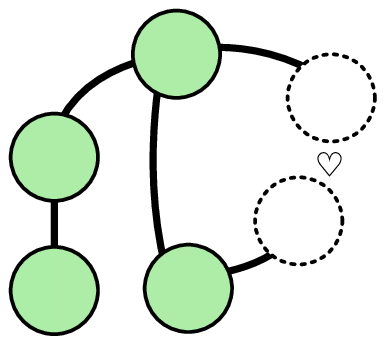}} & 
  \scalebox{0.5}{\includegraphics{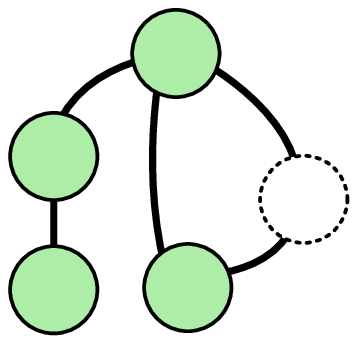}} \\
  $G$ & $relax(G)$ & $H'$ a subgraph of $relax(H)$ & $relax(H)$ & $H$
\end{tabular}
\caption{An illustrative example showing the
  steps involved in the e-graph matching $G \leq_e H$.}
\label{fig:egraphmatching}
 \end{figure}

\subsection{Composing Graphs}

We now introduce a general method for composing graphs which we call
\emph{plugging}; it works equally well for graphs, open graphs, and
e-graphs.  We will give here the definition for the case of e-graphs,
but the reader will have no difficulty in modifying the definitions
accordingly. 

Let $G$ be a graph (not open or extended), and suppose that  we have a
partition of its vertices $V_G = F + B$ into a \emph{front set} and a
\emph{back set}; in this case call the  triple $(G,F,B)$ a
\emph{two-sided graph}.  

\begin{definition}
  A two sided e-graph, $(\pi, V_\pi, F, B)$, with a pair of embeddings
  $p_1$ and $p_2$ is said the be the plugging of two e-graphs $(G,X)$
  and $(H,Y)$, when $\pi \leq_e G$ and $\pi \leq_e H$ by $p_1$ and
  $p_2$ respectively, such that $p_1(F) \subseteq X$ and $p_2(B)
  \subseteq Y$.
  Then we define the \emph{plugging}, $\pi_{p}(G,H)$, as via the
  pushout:
  \begin{diagram}
    \pi & \rInto^{p_1} & G & \\
\dInto<{p_2} & & \dTo & \\
H & \rTo{}  & \NWpbk \pi^{p_1}_{p_2}& \hspace{-0.5cm} (G,H)
  \end{diagram}
  We let $\pi(G,H)$ abbreviate $\exists
  p_1,p_2.\,\pi^{p_1}_{p_2}(G,H)$. An example illustrating plugging is
  given in Figure~\ref{fig:plugging}.
\end{definition}

\begin{figure}[t]
  \centering{\scalebox{0.5}{\includegraphics{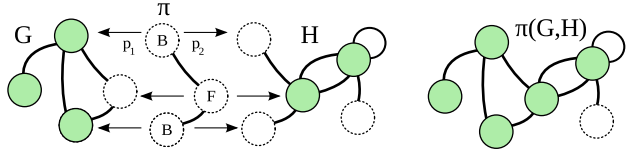}}}
\caption{The plugging of G and H via the 
    two-sided e-graph $\pi$ with embeddings $p_1$ and $p_2$.}
\label{fig:plugging}
\end{figure}

\begin{proposition}
  Let $\pi$, $G$, and $H$ be as above, and let $K$ be some e-graph; then 
  \begin{itemize}
  \item $\pi(G,H) \iso \pi(H,G)$;
  \item $G \leq_e \pi(G,H)$ and $H \leq_e \pi(G,H)$;
  \item $K \leq_e G$ implies $K \leq_e  \pi(G,H)$;
  \end{itemize}
\end{proposition}

\section{Compact Closed Categories}
\label{sec:comp-clos-categ}

\begin{definition}
\label{compactcat-def}
A strict symmetric monoidal category~\cite{AspLon:CatTypStruct:1991}
is called compact closed~\cite{KelLap:comcl:1980} when each object $A$
has a chosen dual object $A^*$, and morphisms
\begin{gather*}
  d_A : I \to A^* \otimes A \quad\quad\quad e_A : A \otimes A^* \to I
\end{gather*}
where $I$ is the tensor identity of the compact closed category, such
that
\begin{align}
  A \iso A \otimes I \rTo^{\id{A} \otimes d_A} A \otimes A^* \otimes A
  \rTo^{e_A \otimes \id{A}} I \otimes A \iso A & = \id{A} \label{eq:comcl1}\\
  A^* \iso I \otimes A^* \rTo^{ d_A \otimes \id{A^*}} A^* \otimes A
  \otimes A^* \rTo^{\id{A^*} \otimes e_A} A^* \otimes I \iso A^* & =
  \id{A^*} \label{eq:comcl2}
\end{align}
\end{definition}

Every arrow $f:A\to B$ in a compact closed category \catC
has a \emph{name} and \emph{coname}:
\[
\name{f} : I \to A^* \otimes B, \qquad \coname{f} : A \otimes  B^* \to I,
\]
which are constructed as $\name{f} = (\id{A^*}\otimes f) \circ d_A$ and
$\coname{f} = e_B \circ (f \otimes \id{B^*})$.  Hence there are natural
isomorphisms $\catC(A,B) \iso \catC(I,A^*\otimes B) \iso
\catC(A\otimes B^*,I)$ making \catC monoidally closed\footnote{In
  general compact closed categories  are models of multiplicative
  linear logic where $A \multimap B$ is defined as $A^\bot \otimes B$.}.
Furthermore,  $f$ has a dual, $f^* : B^* \to A^*$, defined by 
\[
f^* = (\id{A^*} \otimes e_B) \circ (\id{A^*}\otimes f \otimes
\id{B^*}) \circ (d_A \otimes \id{B^*})
\]
By virtue of equations \eqref{eq:comcl1} and \eqref{eq:comcl2}, $f^{**} =
f$.  Thus $(\cdot)^*$ lifts to an involutive functor
$\catC^{\text{op}} \to \catC$,  making $\catC$ equivalent to its
opposite.

\subsection{Graph Representations for Compact Closed Categories}
\label{sec:graph-repr-comp}

Open graphs with certain additional structure give a representation for
compact closed categories; we now give an overview of this
construction.  The details omitted here can be found in
\cite{Duncan:thesis:2006}.  Pictorial representations are in
Fig.~\ref{fig:comcl-graphs}. We make the convention that the domain of an
arrow is at the top of the picture, and its codomain is at the bottom.

\begin{figure}[t]
  \centering
  \[
  \begin{array}{lllll}
      \id{A \otimes B^*} = \inlinegraphic{3em}{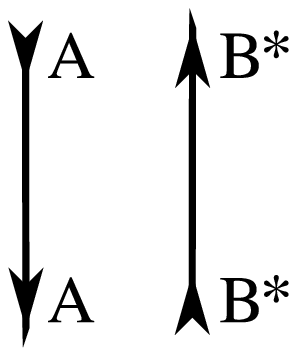}
      &\qquad&
      d_A = \inlinegraphic{2em}{}
      &\qquad& 
      e_B = \inlinegraphic{2em}{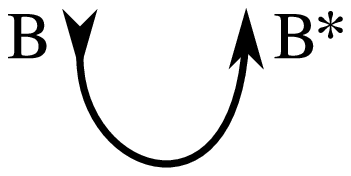} 
      \\\\
      f = \inlinegraphic{3em}{} 
      & \qquad &
      \name{f} = \inlinegraphic{3em}{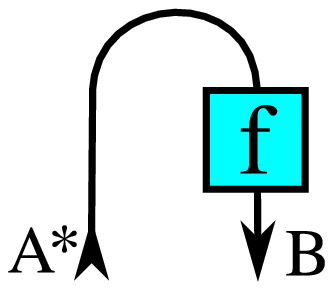}
      &\qquad &
      f^* =  \inlinegraphic{3em}{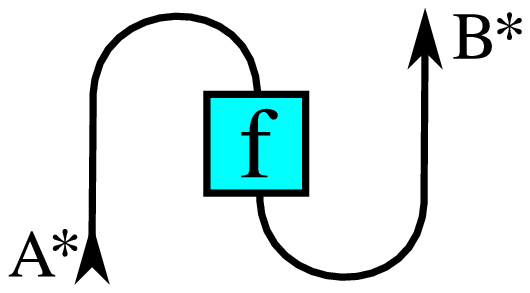}
  \end{array}
  \]
  \caption{Compact Closed Structure as Graphs. }
  \label{fig:comcl-graphs}
\end{figure}

A \emph{concrete graph} $\Gamma$ is 5-tuple $(G, \dom\Gamma, \cod\Gamma,
<_{\text{in}(\cdot)}, <_{\text{out}(\cdot)})$ where:
\begin{itemize}
\item $G = (V,E,s,t)$ is a graph;
\item $\dom\Gamma$ and $\cod\Gamma$ are totally ordered disjoint sets of
  degree one vertices of $G$.  Therefore the union of these sets is the
  \emph{boundary} of the open graph $(G, \dom\Gamma, \cod\Gamma)$;
\item $<_{\text{in}(\cdot)}$ is a family of maps, indexed by $V$ such
  that $<_{\text{in}(v)} : \text{in}(v) \rTo^\isomorphism
  \mathbb{N}_k$ where $k = \sizeof{\text{in}(v)}$.
\item $<_{\text{out}(\cdot)}$ is a family of maps, indexed by $V$ such
  that $<_{\text{out}(v)} : \text{out}(v) \rTo^\isomorphism
  \mathbb{N}_{k'}$ where $k' = \sizeof{\text{out}(v)}$.
\end{itemize}

Since the sets $\dom\Gamma$ and $\cod\Gamma$ consist of vertices of degree
one, we can assign a polarity to each one:  $v \mapsto +$ if the edge
incident at $v$ is an incoming edge; $v \mapsto -$ otherwise.  Hence
$\cod \Gamma$ and $\dom \Gamma$ are \emph{ordered signed sets}.  Given any
ordered signed set $S$ we write $S^*$ for the same ordered set with
the opposite signing.   Given two such sets we can define their disjoint
union $R+S$ as the disjoint union of the underlying sets, inheriting
the signing and the order from $R$ and $S$, with the convention that
$r < s$ for all $r\in R, s\in S$.

\begin{proposition}
Concrete graphs form a compact closed category whose objects are
ordered signed sets and whose arrows $f:A\to B$ are concrete graphs with
$\cod f = B$  and $\dom f = A^*$.
\end{proposition}
For each ordered signed set $A$, the identity map $\id{A}$ has $\dom
\id{A} = A^*$ and $\cod \id{A} = A$; its underlying graph has $E = A$
and $V = A^* + A$ with $t(a) = a$ and $s(a) = a^*$.  Given a pair of
concrete graphs $f:A\to B$ and $g:B\to C$ their composition $g\circ
f:A\to C$ is constructed by merging the two graphs, erasing the
vertices of $\cod f$ and $\dom g$ (called the \emph{boundary
  vertices}), and identifying the edges previously incident at the
deleted vertices.  (Due to the opposite polarity of the domain and
codomain the edges have compatible direction.)  The tensor product on
objects $A,B$ is simply $A+B$; given $f:A\to B$, $g: C\to D$, the
graph of $f \otimes g$ is the disjoint union of the graphs of $f$ and
$g$.  The unit for the tensor is the empty set.  The morphisms $d_A :
I \to A^* \otimes A$, $e_A : A \otimes A^* \to I$ have the same
underlying graph as $\id{A}$, but $\dom d = \emptyset$, $\cod d = A^*+A$, $\dom e
= A+A^*$ and $\cod e = \emptyset$.

\begin{remark}
Although we have not written it explicitly, both composition
and tensor can both be expressed as plugging.  The tensor is the
plugging along the empty graph, while composition is plugging along
an identity graph.  In fact, one can define another compact closed
category  whose objects are two-sided graphs and whose arrows are
e-graphs; sadly, space does not allow it to be described here.
\end{remark}

This category captures exactly the axioms for compact closed
structure, in the sense that any freely generated compact closed
category can be represented by concrete graphs.  We will consider
a collection of basic terms\footnote{See \cite{Duncan:thesis:2006} for
  a more thorough description of the nature of the terms.} $F$
whose types are vectors of some set of basic types $T$.  Then:

\begin{definition}
  A \emph{$T,F$-labelling} $\theta$ for a concrete graph $\Gamma$ is a pair of
  maps  $\theta_T : E \to T$ and $\theta_F : (V - \cod\Gamma -
  \dom\Gamma) \to F$  such that for each vertex  $v$, if
  $\text{in}(v) = \langle a_1, \ldots, a_n\rangle$ and $\text{out}(v)
  = \langle b_1, \ldots, b_m\rangle$ then 
  \[
  \theta v : \langle \theta a_1, \ldots, \theta a_n \rangle
  \to 
  \langle \theta b_1, \ldots, \theta b_m \rangle
  \]
  We say a concrete graph $\Gamma$ is \emph{$T,F$-labellable} if there exists an 
  $T,F$-labelling for it; and if $\theta$ is a labelling for $\Gamma$, then
 the pair $(\Gamma,\theta)$ is called a \emph{$T,F$-labelled graph}.
\end{definition}

The $T,F$-labelled graphs form a compact closed category in the same
way as the concrete graphs, subject to the further restriction
that arrows are composable only when their labellings agree.  

\begin{theorem}
  Let \catC be a compact closed category, freely generated by some set
  of arrows $F$ and ground types $T$;  then \catC is equivalent to the
  category  of $T,F$-labelled graphs.
\end{theorem}

Given a compact closed  category  \catC generated by some basic set of
operations,  the arrows of \catC have a canonical representation as
labelled graphs.  A consequence of the theorem is then that two arrows
are equal by the equations of the compact closed  structure if and
only if their graph representations are equal.

As a final remark before moving on, note that the external structure
of a vertex in a concrete graph is essentially the same as that of a
complete graph; hence one can consistently view subgraphs as vertices,
and abstract over the their internal structure.

\section{!-Boxes}

To support reasoning with spiders we introduce the operation
\emph{!-boxing} (pronounced bang-boxing), on graph representations.
Given a graph representation, this introduces a new notation, that of
outlining a set of nodes (!-boxing them). We then introduce matching
which formalises the idea that a !-box graph can have an arbitrary
number of copies of the !-boxed nodes where every copy connects in the
same way to the nodes outside the !-box.




\begin{definition}[!-box graph]
  A \emph{!-box graph} is a pair $(G, {\cal B}_G)$ where $G$ is a
  graph and ${\cal B}_G$ is the graphs !-boxes which are is a set of
  disjoint subsets of $V_G$.\footnote{One could consider more
    expressive notions of nested, or overlapping, node sets in the
    !-boxes. While such expressivity is interesting, it is not
    required for the system we formalise here.} 
\end{definition}


\begin{definition}[!-box matching]
  we write $(G,{\cal B}_G) \leq_! (H, {\cal C})$, for $(G,{\cal
    B}_G)$ matches $(H, {\cal B}_H)$. This is a binary relation such
  that $(H, {\cal B}_H)$ can be obtained from $(G,{\cal B}_G)$ by the
  following operations on graphs, performed in order:

%
\begin{description}

\item[{\bf copy($c$,$(G,{\cal B}_G)$)}]: the function $c$ is mapping
  from ${\cal B}_G$ to natural numbers. Each bang box, $b$, is copied
  $c(b)$ times. Any edges between a node, $n$, inside a !-box $b$, and
  a node, $m$, outside it, get copied so that there is a new edge from
  $m$ to the new copy of $n$. When $c(b) = 0$, we call it
  \emph{killing} as all nodes in the !-box get removed with any
  incident edges. When $c(b)$ is 1, no additional copies are made and
  we allow ourself to omit this case when writing the function.

\item[{\bf drop($K$,$(G,{\cal B}_G)$)}]: removes the subset, $K$ of the
  !-boxes, but leaves their contents in the graph.
 
\item[{\bf merging($M$,$(G,{\cal B}_G)$)}]: given the set of disjoint
  subsets of unconnected !-boxes, $M$, merging simply unions the
  members. 

\end{description}

\noindent An illustration of matching with these operations is given
in Figure~\ref{fig:bang-box-example}.
\end{definition}

\begin{proposition}[!-box matching is partial order]
\label{thm:bang-box-po}
Reflexivity comes from the trivial matching (no killing, no copying,
no dropping and no merging). Transitivity can be proved by
constructing combined matching from two existing matches: killing a
!-box that was constructed from a copy simply avoids copying the !-box
in the first place, copying after merging simply involves additional
copying beforehand and merging at the end. Antisymmetry can be proved
by constructing from arbitrary matches $G \leq_!  H$ and $H \leq_! G$
the trivial matching. The construction simply involves removing any
killings.
\end{proposition}

\begin{figure}[t]
  \scalebox{0.8}{
    \begin{minipage}[c]{1.0\linewidth}
      \[
      \begin{array}{ccccccc}
        \inlinegraphic{4em}{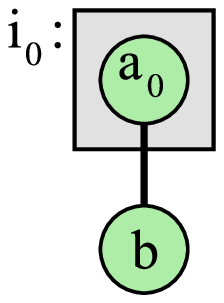}
        & \leq &
        \inlinegraphic{4em}{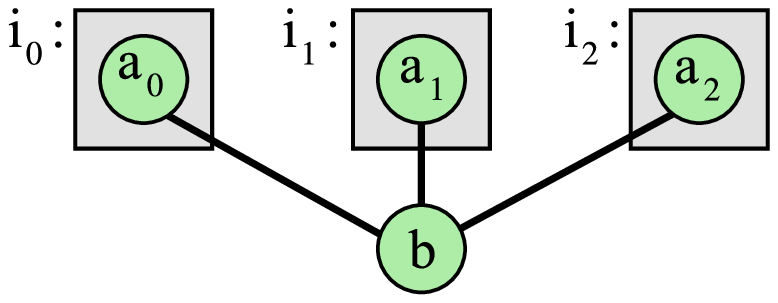} 
        & \leq &
        \inlinegraphic{4em}{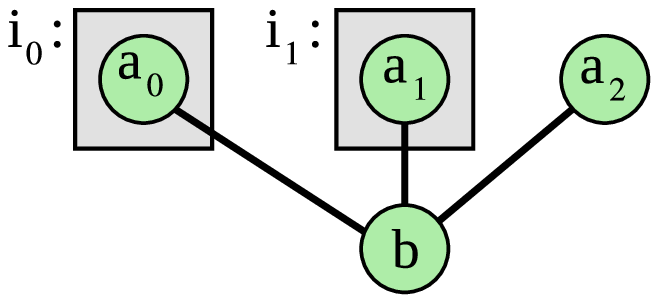}
        & \leq& 
        \inlinegraphic{4em}{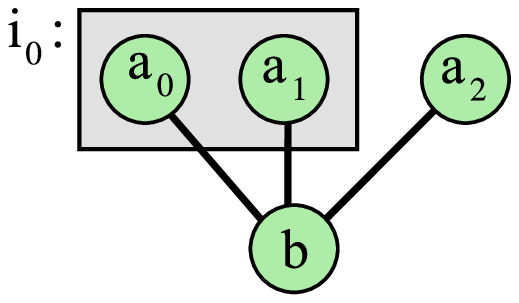} \\
        \stackrel{\ }{G_1}
        & & 
        \stackrel{\ }{G_2 = copy(\{i_0\rightarrow2\},G_1)}
        & &
        \stackrel{\ }{G_3 = drop(\{i_2\}, G_2)}
        & & 
        \stackrel{\ }{merge(\{\{i_0,i_1\}\}, G_3)}
      \end{array}
      \]
    \end{minipage}
}
\caption{An illustration of !-box graph
  matching using the !-box operations. This involves first copying
  !-box $i_0$ twice, then merging $i_0$ and $i_1$ and finally dropping
  $i_3$.}
\label{fig:bang-box-example}
\end{figure}
 

We give a formal semantics to !-box graphs in terms of a set of graphs
in the underlying representation. In particular, we denote the
interpretation of a !-box graph $(G,{\cal B})$ by $\binterp{(G,{\cal
    B})}$ and say that its members are instances.

\begin{definition}[!-box Interpretation]
  $\binterp{(G,{\cal B})}$ is the set of graphs matched by the !-box
  graph that have no !-boxes: $\binterp{(G,{\cal B})} = \{ H \;\;|\;\; (G,{\cal B}) \leq_! (H,\emptyset) \}$
\end{definition}

Observe that every instance of a !-box graph can be defined
by pairing each !-box with the natural number that defines how many
copies are made of it. Thus $\binterp{G}$ is isomorphic to the set of
$k$-tuples of natural numbers, where $k$ is the number of !-boxes.
The need for the !-box matching operation, rather than using a direct
$k$-tuple interpretation, is to allow matching between !-box graphs,
and thus to provide a mechanism for derived rules.

\begin{proposition}
\label{thm:bang-box-respect}
\emph{!-Matching respects !-box semantics}: $G \leq_! H
\Leftrightarrow \binterp{G} \supseteq \binterp{H}$. The proof is a
simple consequence from the $\leq_!$ being a partial order and the
definition of $\binterp{G}$ being a subset of the graphs that match
$G$.
\end{proposition}







Because !-box graphs correspond to a countably infinite number of
concrete graphs, matching cannot be implemented by simply unfolding
all interpretations. We now prove that matching is still decidable.

\begin{theorem}\label{thm:bang-box-decideable-matching}
  \emph{!-box graph Matching is decidable}. The key observation is
  that a graph, $G$, will never match a graph with fewer nodes except
  by killing. Thus the copying(and killing) operations on $G$ can be
  bounded by the number of nodes in the graph it is being matched
  against. While this gives a generate and test style algorithm, it is
  not efficient.  The intuition for an efficient algorithm is to
  search through $G$ incrementally increasing the matched part.
\end{theorem}

\section{Reasoning with Graph Patterns}
\label{sec:rewriting}
\label{sec:patterns}


The representation formed by adding !-boxes to e-graphs, which we call
\emph{graph patterns}, allows us to express, in a finite way, certain
infinite families of equations between e-graphs. In particular, the
Spider Theorem can now be represented as shown in
Figure~\ref{fig:formal-spider}. We now define graph patterns and then
describe how they can be used to develop a formal system for reasoning
about compact closed categories. 

\begin{figure}[t]
$$\begin{array}{ccc}
\inlinegraphic{8em}{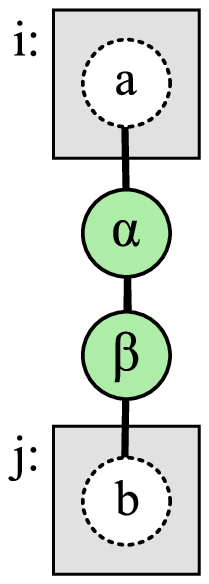} & \;=\; & \inlinegraphic{8em}{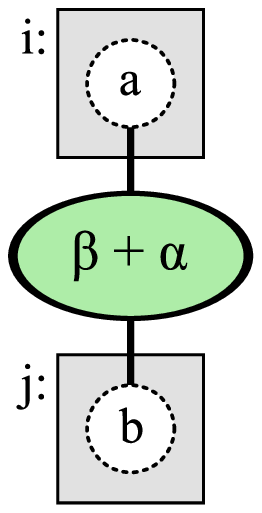}
\end{array}$$
\caption{The Spider Theorem, from \S\ref{sec:informal-spider}, expressed formally using graph
  patterns. The !-boxes are named $i$ and $j$.  The variable nodes are
  white and named $a$ and $b$. The non-variable node data (the angle)
  is written inside the node when non-zero. }\label{fig:formal-spider}
\end{figure}


\begin{definition}[Graph Pattern] A graph pattern $G$ is a !-boxed
  e-graph, i.e. a pair $(eg(G), {\cal B}_G)$, where $eg(G)$ is an
  e-graph and ${\cal B}_G$ is the graph patterns !-boxes.
\end{definition}

\begin{definition}[Graph Pattern Interpretation] A graph pattern $G$
  represents the set of open concrete graphs: $\minterp{G} =
  \bigcup\{\vinterp{G'}\;.\; G' \in \binterp{G}\}$
\end{definition}

This definition allows us to lift matching and plugging from e-graphs
through !-boxes to develop analogous definitions for graph patterns.

\begin{definition}[Graph Pattern Matching] We write $G \leq_p H$ for
  graph pattern $G$ matches graph pattern $H$ and define it to be:

$$G \leq_p H = \exists G'.\, G \leq_! G' \land G' \leq_e H$$
\end{definition}

\begin{definition}[Graph Pattern Plugging] Plugging together of two
  graph patterns is an extension of e-graph plugging where !-boxes
  membership is respected. It restricts identification of nodes in a
  plugging to cases when their !-boxes are can also be identified:
  $\pi_{p_1}^{p_2}(G_1,G_2) = H$ only when, if $b_i \in {\cal
    B}_{G_i}$ and $\exists v.\;p_i(v) \in b_i$ then $\forall w \in
  b_i.\; \exists v'.  p_i(v') = w \land \exists b_j \in {\cal
    B}_{G_j}.\; p_j(v') \in b_j$. This identifies the !-boxes $b_i$
  and $b_j$ which must each come from a distinct one of $G_1$ or
  $G_2$.

\end{definition}

These definitions allow the properties of plugging for e-graphs to
lift naturally to graph patterns. 

The language of graph patterns forms a \emph{meta-level} framework for
reasoning about compact closed categories. The meta-level provides
generic machinery to manipulate graphs and derive new rules.
Following the terminology of logical frameworks, such as
Isabelle~\cite{isabelle}, we call specification of additional
structure, beyond the meta-level, the \emph{object-level}. In our
setting, this involves providing a set of equations between graphs.
These equations are the axioms for the object system. For example, in
\S\ref{sec:case-study} we define an object level theory for reasoning
about quantum computation based on the graphical calculus introduced
in \S\ref{sec:exampl-quant}. In addition to the axioms, the object
level can also provides an appropriate matching or unification operation
for data in the nodes and edges.

We now describe the meta-level framework, noting the conditions for a
rule to be valid, and prove the systems adequacy. The resulting system
forms the basis for an interactive proof assistant that supports
reasoning compact closed categories.

\subsection{Equational Rules}

In our framework, the axioms defined by an object-level model, as well
as derived rules, are pairs of graph patterns. The elements of the
pair represent the left and right hand sides of an equation. Rules are
declarative in that they denote a set of equations between the
underlying formalism of concrete graphs.

The intuitive idea of substitution with a rule is to replace a
subgraph that matches the left hand side with the rule's right hand
side.  However, not all pairs of graphs make a valid rule with respect
to the underlying semantics. For an equation to be well defined with
respect to the compact closed structure it must not be possible to
change the type (the boundary nodes in the domain and co-domain) of an
concrete graph graph by rewriting.  Mapping this restriction back to
rules on graph pattern results in the following conditions:

\begin{itemize}

\item There has to be a isomorphism between exterior nodes in the left
  and right hand sides.

\item Rules must also define a partial mapping between
  !-boxes on the left and right hand sides. The intuition for this
  mapping is that the unfolding used when matching a !-box on the
  left, is applied to the mapped !-box on the right before
  replacement.

\item When an exterior node appears within a !-box on one side of a
  rule, it must also appear under a mapped !-box on the other side.

\end{itemize}

For notational convenience, we annotate !-boxes and exterior nodes in
a graph with unique names. For example, see Figure~\ref{fig:formal-spider}
which shows the Spider Theorem, where the mapping between !-boxes is
represented by !-boxes having the same. Similarly, the isomorphism
between exterior nodes is captured by the set of exterior node names
being equal.

\subsection{Meta-Level Logic and Derived Rules}

Having defined what makes a valid rule, we now present the meta-logic
of the framework. This is quite simple as it only involves dealing
with equations:

\begin{center}
\prooftree
A = B \in \Gamma
\justifies
\Gamma \vdash A = B
\using\mbox{trivial}
\endprooftree
\quad\quad
\prooftree
\justifies
\Gamma \vdash A = A
\using\mbox{refl}
\endprooftree 
\quad\quad
\prooftree
\Gamma \vdash A = B
\justifies
\Gamma \vdash B = A
\using\mbox{sym}
\endprooftree
\begin{center}
\end{center}
\prooftree
\Gamma \vdash A = B \quad\quad
\Gamma \vdash C = D
\justifies
\Gamma \vdash (C = D[A/B])\theta
\using\mbox{subst}
\endprooftree
\quad\quad
\prooftree
\Gamma \vdash A = B \quad\quad
\Gamma \vdash C = D
\justifies
\Gamma \vdash \pi(A,C) = \pi(B,D)
\using\mbox{plug}
\endprooftree
\end{center}

\noindent where $\Gamma$ is the set of object-level axioms, $D[A/B]$
is the graph $D$ with a matching of $A$ replaced by $B$, and $\theta$
is a matching or unification result defined by object level matching
for the node and edge data.

For the reflexivity rule (\emph{refl}), we assume that $A$ is a
well-formed pattern graph. This rule allows a new graph to be
introduced. The \emph{plug} rule allows graphs to be put together to
form larger graphs by plugging in an analogous way to composition in
functional programming. By then applying the \emph{subst} rule,
intermediate results are derived which can themselves be used to
rewrite other rules and conjectures. Given that the axioms in $\Gamma$
meet the validity conditions described earlier, the rules all preserve
the validity of equations and thus the system as a whole ensures only
valid rules are derived.

Given an object-level formalism, a set of equations can be applied
automatically to simplify a graph or simulate computation. For such
rewriting to terminate, a suitable left-to-right ordering on rules
needs to be observed, such as a decrease in the size of the graph. An
initial study into this issue has been investigated by
Kissinger~\cite{Kissinger08Msc}. An example of simulating a quantum
computation is given in~\S\ref{sec:case-study}.

\subsection{Lifting Axioms and Adequacy} 

The axioms of an object formalism come from the semantics of the
underlying system. For instance, the equations given in
Figure~\ref{fig:graph-quant-eqns} can be proved by matrix calculations
in the underlying model. When such rules are expressed as graph
patterns, we replace the concrete representation's boundary nodes with
exterior nodes. This operation is called \emph{lifting}. When a rule
contains exterior nodes, the equation on graph patterns corresponds to
an infinite family of equations between concrete graphs. Thus we might
worry that the lifted equations express too much: they may allow
rewrites which are not true. We call the property that the lifted
representation is a conservative extension of the initial theory
\emph{adequacy}. For models of compact closed categories, the proof of
adequacy is quite simple: given an equation between concrete graphs,
$G = H$, we observe that every instance of the lifted equation has a
subgraph matching the original equation such that the instance can be
derived by plugging.

\section{A Case Study in Quantum Computation}
\label{sec:case-study}

The model of quantum computation introduced in
\S\ref{sec:exampl-quant} provides an object formalism for our
meta-level framework. In particular, the object level axioms come from
lifting the equations in Figure~\ref{fig:graph-quant-eqns} and from
the formalisation of the Spider Theorem in
Figure~\ref{fig:formal-spider}.  Our model of quantum computation
requires no data for the edges. The nodes on the other hand are either
$H$ (a Hadamard gate) with no additional data, or a Pauli operator
which has an \emph{angle} and a \emph{colour}. The colour is red
operations on the $X$ basis and a lighter green for those the $Z$
basis. For their part, angles are expressed as rational numbers which
correspond to the coefficient of $\pi$ in the underlying matrix.

To allow composition of rules to compute the resulting angles we give
the $X$ and $Z$ nodes an \emph{angle expression}. When a node is
within a !-box, the expression is a single \emph{angle-variable} which
gets instantiated to a new angle-variable in each of the unfoldings of
the !-box.  When a node is not within a !-box, the angle-expression is
a mapping from a set of angle-variables to the corresponding rational
coefficient. When an angle-expression contains an angle-variable
within a !-box, this is interpreted as a sum of the variables that
result from its unfolding. 

This rather simple expression language has a normal form by ordering
the angle-variable by name.  Matching then results in angle-variables
being instantiated and the expressions in all affected nodes are then
(re)normalised. An additional implementation detail must also be
observed for the substitution rule: it must ensure that
angle-variables in the rule being applied are distinct from those in
the expression being rewritten.

The quantum Fourier transform is among the most important quantum
algorithms, forming an essential part of Shor's
algorithm~\cite{Shor:PolyTimeFact:1997}, famous for providing
polynomial factoring. In our graph pattern calculus this circuit
becomes the top-left graph in Figure~\ref{fig:quantum-transform}. This
figure shows how computation can be symbolically performed by
rewriting with the lifted equations from
Figure~\ref{fig:graph-quant-eqns} and the graph pattern version of the
Spider Theorem.


\begin{figure}[t]
\begin{tabular}{cccc}
& \inlinegraphic{1.8cm}{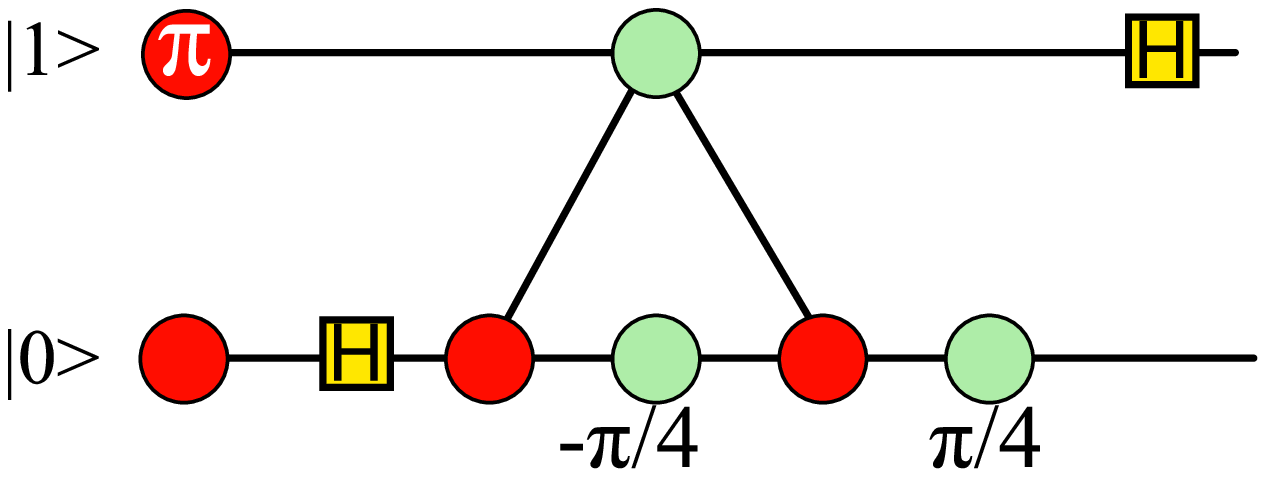} &=& \inlinegraphic{1.8cm}{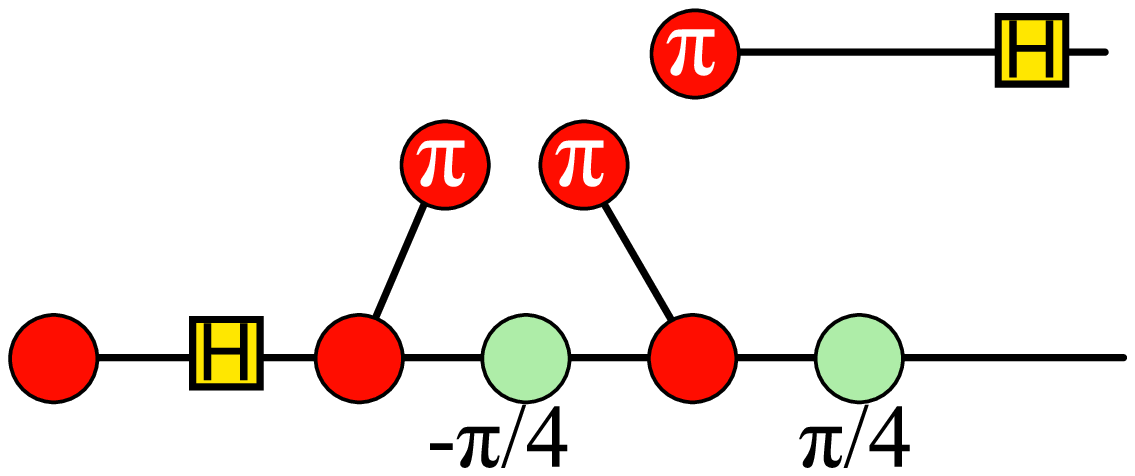} \\ 
& & & \\
& & & \\
=& \inlinegraphic{1.5cm}{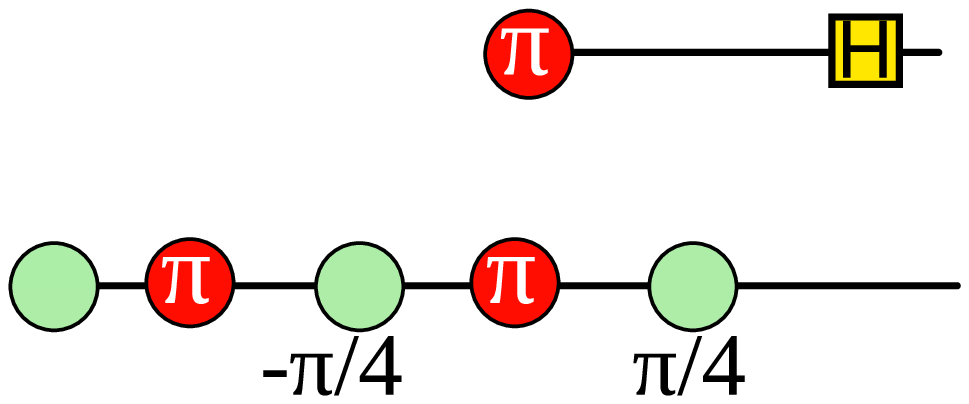} &=& \inlinegraphic{1.5cm}{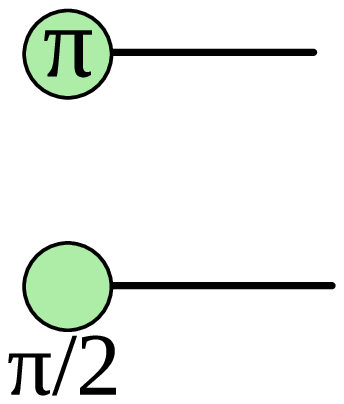} \\
\end{tabular}
\caption{An example computation of the
  Quantum Fourier Transform with inputs 1 and 0, performed
  symbolically by rewriting. }
\label{fig:quantum-transform}
\end{figure}

\section{Related Work}
\label{sec:relatedwork}

There are several foundational approaches to graph transformation,
including algebraic approaches~\cite{corradini97algebraic}, node-label
controlled~\cite{Graphgrammars83}, matrix
based~\cite{DBLP:conf/gg/VelascoL06}, and programmed graph
replacement~\cite{progress97}. These provide general ways of
understanding graph transformations which can then be implemented to
provide machinery for a specific application.  However, systems based
on these theories do not provide machinery for the semantics of
compact closed categories. The distinctive feature of our form of
graph rewriting is that the graphs capture the structural properties
of compact closed categories and our formalism provides a
compositional form of rewriting is: it preserves the type of the
rewritten subgraph.  This allows us to define a plugging operation
over which rewriting distributes.

Bundy and Richardson have described an account of ellipses notation
for lists~\cite{Bundy99proofsabout}. Various authors have also
considered ellipses representations for
matrices~\cite{SextonSorge05,Pollet04intuitiveand}, and more recently,
Prince, Ghani and McBride have developed a general formalism for
ellipses notation using Containers~\cite{conf/flops/PrinceGM08}.
Providing machinery for rewriting of graphs with ellipses notation,
which is needed to represent the Spider Theorem, is a novel
contribution of our approach to graph rewriting.

We note that our graphical notation has little connection
to \emph{graph states} as used in various approaches to
measurement-based quantum computation \cite{Raussendorf-2001}.  In
that approach the graph structure is used to provide a description of
the entanglement in a state:  it does not provide a complete
description of a computation.



\section{Conclusions and Further Work}
\label{sec:conclusions}


We extended the representation of compact closed categories as graphs
to provide a more expressive account of the interface offered by an
open graph. This representation enjoys a plugging operation that has
sequential and parallel composition as special cases. We also
described a formalism for ellipses notations on graphs and showed that
matching is decidable. These representations, together, provide a rich
language of \emph{graph patterns}. This provides the foundation for a
simple meta-logic for reasoning about models of compact closed
categories. 

We use the graphical language extend existing graphical calculi for
quantum computation. In particular, informal reasoning with graphical
equations that contain ellipses notation, such as the Spider Theorem,
now have a formal graphical representation. We illustrate this by
showing how computation can be performed by symbolic graphical
rewriting.

We are left with several exciting avenues for further research.  The
most immediate direction we are pursuing is to provide a full
implementation - only a partial one is currently
available\footnote{\url{http://dream.inf.ed.ac.uk/projects/quantomatic}}.
Other areas of further work include considering confluence results for
sets of rewrite rules, increasing the expressiveness of the
representation for graph-patterns, and finding a complete set of
rewrite rules for the considered model of quantum computation.


%

\bibliographystyle{plain}
\bibliography{bibfile}

\end{document}